\documentclass[letterpaper, 10 pt, conference]{ieeeconf}  

\IEEEoverridecommandlockouts                              

\usepackage{hyperref}
\usepackage{dsfont}
\usepackage[utf8x]{inputenc} 
\usepackage{amsmath}
\usepackage{mathrsfs}
\usepackage{amsfonts}
\usepackage{amssymb}
\usepackage{graphicx}
\usepackage{color}
\usepackage{soul}

\newcommand{\norm}[1]{\left\lVert#1\right\rVert}
\newcommand{\pdev}[2]{\dfrac{\partial #1}{\partial #2}}
\newcommand{\minlamb}[1]{\lambda_{\tt{min}}(#1)}
\newcommand{\maxlamb}[1]{\lambda_{\tt{max}}(#1)}

\newcommand{\inv}[1]{#1^{-1}}

\def\rea{\mathds{R}}

\newtheorem{remark}{Remark}

\newtheorem{Theorem}{Theorem}
\newtheorem{definition}{Definition}

\def\rea{\mathbb{R}}
\def\BibTeX{{\rm B\kern-.05em{\sc i\kern-.025em b}\kern-.08em
		T\kern-.1667em\lower.7ex\hbox{E}\kern-.125emX}}

\title{\Large \bf
Tuning rules for passivity-based integral control for a class of mechanical systems
}

\author{Carmen Chan-Zheng, Mauricio Muñoz-Arias, and Jacquelien M.A. Scherpen
	\thanks{The work of C. Chan-Zheng is sponsored by the University of Costa Rica. The authors are with the Jan C. Willems Center for Systems and Control, and Engineering and Technology institute Groningen (ENTEG), Faculty of Science and Engineering at the University of Groningen,  9747 AG Groningen, The Netherlands (email: c.chan.zheng@rug.nl, m.munoz.arias@rug.nl, j.m.a.scherpen@rug.nl).}}

\begin{document}
\maketitle
\thispagestyle{empty}
\pagestyle{empty}

\begin{abstract}
This manuscript introduces a passivity-based integral control approach for fully-actuated mechanical systems. The novelty of our methodology is that we exploit the gyroscopic forces of the mechanical systems to exponentially stabilize the mechanical system at the desired equilibrium even in the presence of matched disturbances; additionally, we show that our approach is robust against unmatched disturbances. Furthermore, we provide tuning rules to prescribe the performance of the closed-loop system. We conclude this manuscript with experimental results obtained from a robotic arm.
\end{abstract}
\begin{keywords}
	exponential stability, passivity-based control, port-Hamiltonian, tuning
\end{keywords}
\section{Introduction}
	Amidst the modeling approaches for mechanical systems, we find the well-established port-Hamiltonian (pH) framework that provides insight into the roles that the interconnection, the dissipation, and the energy play in the physical system behavior \cite{vanderSchaft2017}. In general, the controllers designed under this framework preserve the pH structure for the closed-loop system, and among these strategies, we find the passivity-based control (PBC) approaches, a well-known set of techniques to control complex physical systems \cite{vanderSchaft2017,ortega2013passivity}. These approaches generally consist of two main steps: i) an energy shaping process that assigns the desired equilibrium to the closed-loop and ii) a damping injection step that ensures that such an equilibrium is asymptotically stable. Some results of PBC approaches for mechanical systems are found in \cite{borja2020new, ortega2004interconnection,gomez2004physical}.
	
	However, steady-state errors usually occur during the practical implementation of these controllers due to external disturbances or unmodeled phenomena. Thus, an integral action is added to the PBC approach to attenuate these disturbances. In \cite{ortega2004interconnection}, we find a passivity-based integral control (PBIC) approach where the integral action is applied directly to the passive outputs (i.e., the velocity for mechanical systems). However, the implementation of this approach is hindered when position control is the main objective as the disturbances may shift the equilibrium or render the system unstable. In \cite{dirksz2010power}, the integral action is applied to both non-passive and passive outputs; however, the controller destroys the pH structure. In \cite{dirksz2011port,donaire2009addition, ortega2012robust,romero2013robust}, we find PBIC approaches that preserve the pH structure but require a change of coordinates. Additional results that do not require a change of coordinates are found in \cite{ferguson2017integral,ferguson2019matched}, but they require a constant mass-inertia matrix and an estimation of the model of the system, respectively. None of the approaches are exponentially stable and lack tuning methodologies to prescribe performance other than stability.
	
	In this paper, we employ the change of coordinates strategy to develop our PBIC approach. The novelty of our technique is that it exploits the gyroscopic forces~--~customarily, these terms are canceled via feedback linearization~--~ of the mechanical system to exponentially stabilize it at the desired equilibrium. Our main contributions are summarized as:
	\begin{itemize}
		\item [(i)] A novel PBIC approach that exponentially stabilizes a class of mechanical systems at the desired equilibrium in the presence of matched disturbances.
		\item  [(ii)] A robustness analysis of our control approach against unmatched disturbances, where we are able to prove that the trajectories of the closed-loop converge exponentially to a small ball around the equilibrium.
		\item[(iii)] Tuning guidelines for prescribing: i) the upper bound of the rate of convergence of the trajectories; ii) the maximum permissible overshoot of the output; and iii) the nonlinear gain margin of the closed-loop system. 
	\end{itemize}
	The remainder of this paper is structured as follows: in Section \ref{sec:pre} we provide the description of the class of mechanical systems under-study, the theoretical backgrounds, and formulate the problem. Section \ref{sec: PBIC_approach} contains the main results of this manuscript, i.e., the development of our PBIC approach. Then, we provide a robustness analysis of our approach against unmatched disturbances in Section \ref{sec:unmatched}. In Section \ref{sec:exp} we apply our control methodology to a three degrees-of-freedom (DoF) robotic arm. We finalize this manuscript with some concluding remarks and future work in Section \ref{sec:concl}.

	\textbf{Notation}: We denote the $n\times n$ identity matrix as $I_n$ and the $n\times m$ matrix of zeros as $0_{n\times m}$. For a given smooth function $f:\rea^n\to \rea$, we define the differential operator $\nabla_x f:=\frac{\partial f}{\partial x}$ which is a column vector, and $\nabla^2_x f:=\frac{\partial^2 f}{\partial x^2}$. For a smooth mapping $F:\rea^n\to\rea^m$, we define the $ij-$element of its $n\times m$ Jacobian matrix as $(\nabla_x F)_{ij}:=\frac{\partial F_i}{\partial x_j}$. When clear from the context the subindex in $\nabla$ is omitted. For a given matrix $A\in\rea^{n\times n}$ and a given vector $x\in\rea^{n}$, we say that $A$ is \textit{positive definite (semi-definite)}, denoted as $A\succ0$  ($A\succeq0$), if $A=A^{\top}$ and $x^{\top}Ax>0$ ($x^{\top}Ax\geq0$) for all $x\in \rea^{n}-\{0_{n} \}$ ($\rea^{n}$). For a given vector $x\in\rea^n$ , we denote the Euclidean norm as $\norm{x}$. For $B=B^\top$, we denote by $\maxlamb{B}$ (resp. $\minlamb{B}$) as the maximum (resp. minimum) eigenvalue of $B$. We denote $\rea_+$ as the set of positive real numbers and $\rea_{\geq0}$ as the set $\rea_+\cup \{0\}$. Let $x,y\in\rea^{n}$, we define $col(x,y):=[x^\top y^\top]^\top$. We denote $e_i$ as the $i^{th}$ element of the canonical basis of $\rea^n$. 

\textbf{Caveat}: when possible, we omit the arguments to simplify the notation. 

\section{Preliminaries and Problem formulation}\label{sec:pre}
This section describes the class of mechanical systems under-study in a pH form. Then, we describe some stabilities concepts and conclude this paper with the problem formulation.
\subsection{Description of a class of mechanical systems}
Consider a standard $n$-DoF mechanical systems given by
	\begin{equation}\label{eq:sigma_PH_mechanical}
		\arraycolsep=1.6pt \def\arraystretch{1.4}
		\begin{array}{rcl}
		\begin{bmatrix}\dot{q}\\\dot{p}\end{bmatrix}	&=&\begin{bmatrix}
				0_{n\times n} &I_n\\
				-I_n & -D\left(x\right)\end{bmatrix}\begin{bmatrix}
 			\nabla_q H(x)\\
 			\nabla_p H(x)\\
 		\end{bmatrix}+\begin{bmatrix}
 		0_{n\times n}\\
 		G\left(q\right)
 		\end{bmatrix}u+\begin{bmatrix}d_u\\d_m\end{bmatrix}
 		\end{array}
	\end{equation}
	with $q,p\in\rea^{n}$ being the generalized configuration coordinates and the generalized momenta vectors, respectively; $x:=col(q,p)$; the fully damped matrix $D:\rea^n\times \rea^n\to\rea^{n\times n}$ which is positive definite; $u,y\in\rea^{n}$ the input and output vectors, respectively; the constant matched and unmatched disturbances $d_m,d_u\in\rea^{n}$, respectively; the input matrix $G\left(q\right)\in\rea^{n\times n}$ everywhere
	invertible, i.e., the pH system is \emph{fully actuated}; and the Hamiltonian of the system $H:\rea^n\times\rea^n\to\rea$ corresponds to
	\begin{equation*}
		H\left(x\right)=\dfrac{1}{2}p^{\top}M^{-1}\left(q\right)p+V\left(q\right)
		\label{eq:pre_H}
	\end{equation*}
	where $M:\rea^n\to\rea^{n\times n}$ is the positive definite
	mass inertia matrix and $V:\rea^n \to\rea$ is the potential energy.
	The set of assignable equilibria for \eqref{eq:sigma_PH_mechanical} is defined by
	\begin{equation*}
		\mathcal{E}:=\{{q},{p} \in \rea^n \ | \ {p}=0_n\}.
	\end{equation*}
    Then, based on the results of \cite{ReyesBaezThesis2019}, the system \eqref{eq:sigma_PH_mechanical} is equivalently rewritten as
	\begin{equation}\label{eq:sigma_PH_mechanical_rewritten}
		\arraycolsep=1pt \def\arraystretch{1.4}
		\begin{array}{rcl}
			\begin{bmatrix}\dot{q}\\\dot{p}\end{bmatrix}&=&
			\begin{bmatrix} 0_{n\times n} &I_n\\-I_n & -E\left(x\right)-D\left(x\right)\end{bmatrix}\begin{bmatrix}
				\nabla_q V(q)\\\nabla_p H(x)\end{bmatrix}\\
			&+& \begin{bmatrix}0_{n\times n}\\G\left(q\right)\end{bmatrix}u
			+\begin{bmatrix}d_u\\d_m\end{bmatrix}\\
			y &=& G\left(q\right)^{\top}\nabla_p H(x)
		\end{array} 
	\end{equation}
	where
	\begin{equation*}
		E\left(x\right)  := 
		S_{H}\left(x\right)
		-
		\dfrac{1}{2}
		\dot{M}\left(q\right)
	\end{equation*}
	with $S_{H}\left(x\right)$ being a skew-symmetric matrix, which the $kj$-th entry ($k,j=1,\hdots,n$) corresponds
	\begin{equation*}
		S_{H_{kj}}(x):=\dfrac{1}{2}\sum_{i=1}^n\left\{\pdev{M_{ki}}{q_j}(q)-\pdev{M_{ij}}{q_k}(q)\right\}e_i\inv{M}(q)p.
	\end{equation*}
	
	Note that $E(x)$ (and consequently $S_H(x)$) are calculated such that
	\begin{equation*}
		\nabla_q\left(\dfrac{1}{2}p^\top M^{-1}\left(q\right)p\right)
		=E\left(x\right) M^{-1}\left(q\right)p.
	\end{equation*}
	The force $E\left(q,p\right)M^{-1}\left(q\right)p$ is partly the Hamiltonian counterpart of the gyroscopic force contained in the Coriolis and the centrifugal terms in $C\left(q,\dot{q}\right)\dot{q}$ in the Euler-Lagrangian framework (see Remark B.7 in Appendix B.2.3 of \cite{ReyesBaezThesis2019} for further details together with the results of \cite{stadlmayr2008tracking} and \cite{sarras2012modeling}).
	The system \eqref{eq:sigma_PH_mechanical_rewritten} is referred as \textit{a mechanical port-Hamiltonian-like system} by \cite{ReyesBaezThesis2019}.

\subsection{Stability properties}
In this manuscript, we consider two stability properties to analyze our PBIC approach in closed-loop with \eqref{eq:sigma_PH_mechanical_rewritten}, that is, exponential stability (ES) and input-to-state stability (ISS). The former ensures that the trajectories of the system under-study are bounded by an exponential decay function; while the latter ensures that the trajectories are bounded for any initial conditions given that the input is also bounded (see \cite{khalil2002nonlinear}). Here, we provide some additional definitions, which we exploit in the sequel for tuning purposes. 

Consider \eqref{eq:sigma_PH_mechanical} with $u=0_n$ and  ${d:=col(d_u,d_m)}$. Moreover, assume that the Hamiltonian $H(x)$ has a local isolated minimum at $x_\star\in\rea^{2n}$, i.e., ${x_\star:=\arg\min H(x)}$, or equivalently, the system \eqref{eq:sigma_PH_mechanical} has a stable equilibrium at $x_\star$. 

First let $d=0_{2n}$, then we introduce the following.
\begin{definition}[Rate of convergence]
	Assume \eqref{eq:sigma_PH_mechanical} is ES. Then. the rate of convergence of \eqref{eq:sigma_PH_mechanical} is the exponential decay value of the trajectories of the system approaching the equilibrium $x_\star$, i.e., there exists some constants $k_1,k_2,k_3\in\rea_+$ such that
	\begin{equation*}
		\norm{x}\leq \sqrt{\dfrac{k_2}{k_1}}\norm{x_0}\exp\left\{-\dfrac{k_3}{2k_2}(t-t_0)\right\}
	\end{equation*}
	with $\dfrac{k_3}{2k_2}$ being the upper bound of the rate of
	convergence, $t_0 \geq 0$ is the initial time, and $x_0$ is the initial condition. See \cite{khalil2002nonlinear} for further details.\hfill$\square$
	~\\
\end{definition}

Now, let $d\neq0_{2n}$, then we provide the following ISS-related feature.
\begin{definition}[Nonlinear stability margin \cite{sontag1995characterizations}]\label{smdef}
	Consider the system \eqref{eq:sigma_PH_mechanical}. Then, the nonlinear stability margin is any function $\rho\in\mathcal{K}_\infty$ that verifies\footnote{We refer the reader to \cite{khalil2002nonlinear} for the definition of $\mathcal{K}_\infty$ and $\mathcal{K}\mathcal{L}$ functions.}
	\begin{equation}\label{smargin}
		\norm{d}\leq\rho(\norm{x}),~\text{and}~\norm{x}\leq\beta(\norm{x_0},t) ~\forall t\geq0,
	\end{equation}
	where $\beta\in \mathcal{K}\mathcal{L}$.
	Moreover, the system \eqref{eq:sigma_PH_mechanical} is said to be ISS if the conditions \eqref{smargin} are satisfied.
	\hfill$\square$
\end{definition}	

\subsection{Problem Formulation}
Propose a control methodology such that it renders the system \eqref{eq:sigma_PH_mechanical_rewritten} exponentially stable at the equilibrium $(q_\star,0_n)$~--~with $q_\star\in\rea^n$ being the desired configuration~--~even in the presence of a vector of constant external disturbances $d_m$ and provide tuning rules to prescribe the performance of the closed-loop system.

\section{An exponentially stable PBIC approach}\label{sec: PBIC_approach}
 Based on the alternative representation of \eqref{eq:sigma_PH_mechanical} as \eqref{eq:sigma_PH_mechanical_rewritten}, we develop a control law such that the mechanical system is rendered exponentially stable\footnote{A similar exponential stability analysis can be found in \cite{chan2021exponential,chan2022tuning}; however, we remark that those papers employ other PBC approaches that do not contain an integral action on the non-passive output.} at the desired equilibrium even under the influence of a matched disturbance, i.e., $d_m\neq0_n, d_u= 0_n$.  
Theorem \ref{Prop1} presents our main result. 

 	\begin{Theorem}\label{Prop1}
 		Consider an alternative pH-like system as \eqref{eq:sigma_PH_mechanical_rewritten} and define
 		\begin{equation*}
 		    \Gamma(x):=(E(x)+D(x)) \inv{M}(q).
 		\end{equation*}
 		Furthermore, consider $K_{p},K_i, K_d, M_d\succ0$; $q_\star \in \rea^{n}$; the change of coordinates
 		$$\bar{q}:=q-q_\star, ~\bar{p}:=p+K_p\bar{q};$$ 
 		and the dynamical extension $z \in \rea^n$ such that
 		\begin{equation*}
 			\dot{z} = -K_{i}\bar{y}
 			\label{eq:z_integral}
 		\end{equation*}
		with 
		\begin{equation}\label{eq:New_Output}
		    \bar{y}:=\inv{M}_d \bar{p}.
		\end{equation}
 		Then, the following statements hold true.
 		\begin{itemize}
 			\item[i)]	The control law 		
				\begin{equation}	\label{eq:Control_Law}
					\arraycolsep=1pt \def\arraystretch{1.4}
					\begin{array}{rcl}
						u&= G\left(q\right)^{-1}\Bigg[&
						\dfrac{\partial V\left(q\right)}{\partial 		q}-M_d\inv{M}(q)K_p\bar{q}-\Gamma(x)K_{p}\bar{q}
						\\&&-K_p\dot{q}-K_d\inv{M}_d\bar{p}+z
						\Bigg],
					\end{array}
				\end{equation}
				renders the system \eqref{eq:sigma_PH_mechanical_rewritten}  \textit{globally exponentially stable} in ${col}\left(q,p\right) = {col}\left(q_\star,0_n\right)$ even in presence of the vector of matched disturbances $d_m$ if and only if
				\begin{equation}\label{eq:cond1}
				\dfrac{1}{2}(\Gamma(x)M_d+M_d\Gamma^\top(x))+K_d\succeq 0.
				\end{equation}
			\item[ii)] The rate of convergence of the trajectories of the closed-loop \eqref{eq:sigma_PH_mechanical_rewritten}-\eqref{eq:Control_Law} is upper bounded by		   
				\begin{equation}\label{eq:rate}
					\dfrac{\mu\beta_{\max}}{1+\epsilon\beta_{\max}\maxlamb{\inv{M}_d}}
				\end{equation}
				for some $\beta_{\max},\epsilon,\mu\in\rea_+$.
			\item[iii)] The maximum overshoot of the output of the system is given by
			\begin{equation}\label{eq:ov}
				\xi:=\maxlamb{\inv{M}_d}\sqrt{\dfrac{\kappa_2}{\kappa_1}}\norm{\bar{x}_0}
			\end{equation}
			for some $\kappa_1,\kappa_2\in\rea_+$ and vector $\bar{x}_0\in\rea^{3n}$.
 		\end{itemize}
 	\end{Theorem}
    ~\\
    \begin{proof}
        To prove i), let $\bar{z}:=(z+d_m)$, then, note that the closed-loop \eqref{eq:sigma_PH_mechanical_rewritten}-\eqref{eq:Control_Law} is again a pH representation given by
        \begin{equation}\label{eq:Closed-loop}
        \arraycolsep=2.4pt \def\arraystretch{1.4}
        \begin{array}{rcl}
            \begin{bmatrix}
                \dot{\bar{q}}\\\dot{\bar{p}}\\\dot{\bar{z}}    
            \end{bmatrix}&=&\begin{bmatrix}
                -\inv{\bar{M}}&\inv{\bar{M}}M_d&0_{n\times n}\\
                -M_d\inv{\bar{M}}&-\bar{\Gamma} M_d-K_d&K_i\\
                0_{n\times n}&-K_i&0_{n\times n}
            \end{bmatrix}\begin{bmatrix}
                \nabla_{\bar{q}} \bar{H}\\
                \nabla_{\bar{p}} \bar{H}\\
                \nabla_{\bar{z}} \bar{H}
            \end{bmatrix}
        \end{array}
        \end{equation}
        with 
        \begin{equation*}
 	    \arraycolsep=1pt \def\arraystretch{1.4}
            \begin{array}{rcl}
                 \bar{M}(\bar{q}):= {M}(\bar{q}+q_\star) , ~
                 \bar{\Gamma}(\bar{q},\bar{p}):=\Gamma(\bar{q}+q_\star,\bar{p}-K_p\bar{q}),
            \end{array}
        \end{equation*}
        the new output as given in \eqref{eq:New_Output}, and the new Hamiltonian
        \begin{equation}\label{eq:New_Hamiltonian}
             \bar{H}(\bar{q},\bar{p},\bar{z})=\frac{1}{2}\bar{p}^\top\inv{M}_d \bar{p}+\frac{1}{2}\bar{q}^\top K_p \bar{q}+\frac{1}{2}\bar{z}^\top\inv{K}_i\bar{z}.
        \end{equation}

      Then, consider the new Hamiltonian as the Lyapunov candidate, which derivative corresponds to
        \begin{equation*}
             \arraycolsep=5pt \def\arraystretch{1.4}
                 \dot{\bar{H}}= -\nabla^\top\bar{H}\left[\begin{array}{rcl}
                 \inv{\bar{M}}&0_{n\times n}\\
                 0_{n\times n}&\bar{\Gamma}M_d+M_d\bar{\Gamma}^\top +K_d
            \end{array}\right]\nabla\bar{H}.
        \end{equation*}
        It follows that $\dot{\bar{H}}\leq 0$ if and only if \eqref{eq:cond1} holds. Thus, 
        the closed-loop \eqref{eq:sigma_PH_mechanical_rewritten}-\eqref{eq:Control_Law} is stable. \textit{Asymptotic stability} follows by invoking LaSalle's principle.  Note that
        \begin{equation*}
        \arraycolsep=1pt \def\arraystretch{1.4}
            \begin{array}{rcl}
                 \dot{\bar{H}}
                 \equiv0&\iff& \bar{q},\bar{p}=0 \implies \dot{\bar{p}}=0 \implies \bar{z}=0.
            \end{array}
        \end{equation*}
    
	    To prove \textit{exponential stability}, let $\bar{x}:=col(\bar{q},\bar{p},\bar{z})$. Then, consider a new Lyapunov candidate, 
	    \begin{equation}\label{eq:Lya}	    	S(\bar{x}):=\bar{H}(\bar{x})-\epsilon\nabla_{\bar{p}}^\top\bar{H}(\bar{x}) {M}_d \nabla_{\bar{z}}\bar{H}(\bar{x})
	    \end{equation}
	    for some\footnote{There always exists a sufficiently small $\epsilon$ such that $S(\bar{x})\in\rea_+$ for all $\bar{x}\neq col(q_\star,0_n,0_n)$.} $\epsilon\in \rea_+$. 
    
	    Note that the new Hamiltonian satisfies
	    \begin{equation}\label{eq:boundsHamiltonian}
	    	\dfrac{\beta_{\min}}{2}\norm{\bar{x}}^2\leq \bar{H}(\bar{x})\leq 	\dfrac{\beta_{\max}}{2}\norm{\bar{x}}^2
	    \end{equation}
	    with 
	    \begin{equation*}
    	\begin{array}{rcl}
    		\beta_{\min}&:=&\min\{\minlamb{{M_d}},\minlamb{K_p},\minlamb{\inv{K_i}}\},\\
    		\beta_{\max}&:=&\max\{\maxlamb{{M_d}},\maxlamb{K_p},\maxlamb{\inv{K_i}}\}.
    	\end{array}
    	\end{equation*}
	    Moreover, by employing Young's inequality, $S-\bar{H}$ satisfies
	    \begin{equation}\label{eq:boundsAdd}
	    	\norm{-\epsilon\nabla_{\bar{p}}^\top\bar{H}{M}_d 	\nabla_{\bar{z}}\bar{H}}\leq\dfrac{\epsilon\beta_{\max}^2\maxlamb{{M}_d}}{2} \norm{\bar{x}}^2.
	    \end{equation}
	    Thus, from \eqref{eq:boundsHamiltonian} and \eqref{eq:boundsAdd}, we get
	    \begin{equation}\label{eq:cond1_exp}
	    	\kappa_1\norm{\bar{x}}^2\leq S(\bar{x})\leq   \kappa_2\norm{\bar{x}}^2
	    \end{equation}
    	with
    	\begin{equation}\label{eq:k1k2}
    		\begin{array}{rcl}
    			\kappa_1&:=&\dfrac{\beta_{\min}-\epsilon\beta_{\max}^2\maxlamb{{M}_d}}{2},\\
    			\kappa_2&:=&\dfrac{\beta_{\max}+\epsilon\beta_{\max}^2\maxlamb{{M}_d}}{2}. 
    		\end{array}
    	\end{equation}
	    Note that there always exists a sufficiently small $\epsilon$ such that $\kappa_1\in\rea_+$.
    
    Then, the derivative of the Lyapunov candidate \eqref{eq:Lya} corresponds to
    \begin{equation}\label{eq:Sdot}
    	\dot{S}=-\nabla^\top \hat{H} \Upsilon \nabla \hat{H}
    \end{equation}
    with
    \begin{equation}\label{eq:Upsilon}
    	\arraycolsep=1pt \def\arraystretch{1.4}
    	\Upsilon(\bar{q},\bar{p}):=\left[ \begin{array}{cc}
    		\Upsilon_{11}&\Upsilon_{12}\\
    		\Upsilon_{12}^\top&\Upsilon_{33}
    	\end{array}\right]
    \end{equation}
    where
    \begin{equation*}
    	\arraycolsep=3pt \def\arraystretch{1.8}
    	\begin{array}{lcl}
    		\Upsilon_{11}&:=&\begin{bmatrix}
    			\inv{\bar{M}}&0_{n\times n}\\
    			0_{n\times n}&\upsilon_{22}
    		\end{bmatrix},~\Upsilon_{12}:=-\dfrac{\epsilon}{2}\begin{bmatrix}\inv{\bar{M}}M_d\\(\bar{\Gamma} M_d+K_d)^\top\end{bmatrix},\\
    		\upsilon_{22}&:=&K_d+\dfrac{1}{2}(\bar{\Gamma} M_d+M_d \bar{\Gamma}^\top)-\epsilon\inv{M_d},\\
    		\Upsilon_{33}&:=&\epsilon K_i.
    	\end{array}
    \end{equation*}
    Recall that \eqref{eq:Sdot} is negative definite if and only if $\Upsilon(\bar{q},\bar{p})\succ0$. 
    To verify the sign of $\Upsilon(\bar{q},\bar{p})$, we employ Schur complement analysis, i.e., there exists a sufficiently small $\epsilon\in\rea_+$ such that $\Upsilon_{11}\succ0$ and
    \begin{equation*}
    	\Upsilon_{33}-\Upsilon_{12}^\top \inv{\Upsilon_{11}}\Upsilon_{12}\succ0;
    \end{equation*}
    thus, $\Upsilon(\bar{q},\bar{p})\succ0$.
    
    Denote with $\mu\in\rea_+$ the minimum eigenvalue of \eqref{eq:Upsilon}, then it follows that
    \begin{equation}\label{eq:cond2_exp}
    	\dot{S}\leq-\mu\norm{\nabla\bar{H}}^2\leq-\mu\beta_{\max}^2\norm{\bar{x}}^2.
    \end{equation}
	Hence, from \eqref{eq:cond1_exp} and \eqref{eq:cond2_exp}, the closed-loop \eqref{eq:sigma_PH_mechanical_rewritten}-\eqref{eq:Control_Law} is exponentially stable at $col(q_\star,0_n)$ (See Theorem 4.10 from \cite{khalil2002nonlinear}).
	
	The global property follows from noting that $S(\bar{x})$ is radially unbounded, i.e.,
	\begin{equation*}
		\norm{\bar{x}}\to\infty\implies S(\bar{x})\to\infty.
	\end{equation*}

	To prove ii), note that from \eqref{eq:cond1_exp} and \eqref{eq:cond2_exp}, we get that
	\begin{equation*}
		\dot{S}\leq -\dfrac{2\mu\beta_{\max}}{1+\epsilon\beta_{\max}\maxlamb{\inv{M}_d}}	S;
	\end{equation*}
	then, it follows that by employing the comparison lemma (see \cite{khalil2002nonlinear}), we get that the trajectories of the closed-loop \eqref{eq:Closed-loop} are upper bounded by an exponential decay function given by
	\begin{equation}\label{eq:boundX}
		\norm{\bar{x}}\leq \sqrt{\dfrac{\kappa_2}{\kappa_1}}\norm{\bar{x}_0}\exp\left\{-\dfrac{\mu\beta_{\max}}{1+\epsilon\beta_{\max}\maxlamb{\inv{M}_d}}t\right\}
	\end{equation}
	with $\bar{x}_0\in\rea^{3n}$ being the initial conditions in the new coordinates, and $\kappa_1,\kappa_2$ are defined in \eqref{eq:k1k2}.
	
	To prove iii), consider \eqref{eq:boundX} and note that \eqref{eq:New_Output} verifies
	\begin{equation*}
	\arraycolsep=1pt \def\arraystretch{1.4}
		\begin{array}{rcl}
			\norm{\bar{y}}&\leq& \maxlamb{\inv{M}_d}	\norm{\bar{x}}\\
			&\leq&  \xi\exp\left\{-\dfrac{\mu\beta_{\max}}{1+\epsilon\beta_{\max}\maxlamb{\inv{M}_d}}t\right\}
		\end{array}
	\end{equation*}
	with $\xi$ defined as in \eqref{eq:ov}.
    \end{proof}
	~\\
	\begin{remark}
		Note that we can employ \eqref{eq:rate} and \eqref{eq:ov} as tuning guidelines as we can see the effect of the control parameters on the upper bound of the rate of convergence and the maximum permissible overshoot, respectively. 
		\hfill$\square$
		~\\
	\end{remark}
   \begin{remark}
		The change of variable $\bar{p}:=p+K_p\bar{q}$~--~introduced first in \cite{donaire2009addition}~--~is also employed in \cite{dirksz2011port,romero2013robust} to develop their respective PBIC approaches, where only asymptotic stability is achieved. The key difference of our approach with the latter methodologies is that we use the gyroscopic related forces (contained in $\Gamma(x)$) as part of the damping injection step to exponentially stabilize \eqref{eq:sigma_PH_mechanical_rewritten} at the desired equilibrium.
		\hfill$\square$
		~\\
	\end{remark}

	We have shown the effectiveness of control law \eqref{eq:Control_Law} against matched disturbances by selecting a convenient choice of Lyapunov candidate. In the sequel, we demonstrate that our approach is also effective against unmatched disturbance. 

\section{Unmatched disturbance attenuation properties of the PBIC}\label{sec:unmatched}
 In this section, we demonstrate that the control law \eqref{eq:Control_Law} also exponentially stabilizes the system \eqref{eq:sigma_PH_mechanical_rewritten} even in the presence of unmatched disturbances, i.e., $d_u\neq0_n$. 
 
 Observe that the closed-loop system \eqref{eq:Closed-loop} becomes
 \begin{equation}\label{eq:Closed-loop2}
 	\arraycolsep=1.8pt \def\arraystretch{1.4}
 	\begin{array}{rcl}
 		\begin{bmatrix}
 			\dot{\bar{q}}\\\dot{\bar{p}}\\\dot{{z}}    
 		\end{bmatrix}&=&\begin{bmatrix}
 			-\inv{\bar{M}}&\inv{\bar{M}}M_d&0_{n\times n}\\
 			-M_d\inv{\bar{M}}&-\bar{\Gamma} M_d-K_d&K_i\\
 			0_{n\times n}&-K_i&0_{n\times n}
 		\end{bmatrix}\begin{bmatrix}
 			\nabla_{\bar{q}} \hat{H}\\
 			\nabla_{\bar{p}} \hat{H}\\
 			\nabla_{z} \hat{H}
 		\end{bmatrix}+\begin{bmatrix}
 			d_u\\0_n\\0_n
 		\end{bmatrix}
 	\end{array}
 \end{equation}
with Hamiltonian defined as in \eqref{eq:New_Hamiltonian}.
Then, we exploit the same Lyapunov candidate \eqref{eq:Lya} to perform an ISS analysis on \eqref{eq:Closed-loop2}. 
The analysis follows a similar method as \cite{chan2022tuning} for another closed-loop system.
\begin{Theorem}\label{Prop2}
	Let $d_u\neq0_n$ and consider the control law \eqref{eq:Control_Law}. Then, the following statements are hold true.
	\begin{itemize}
		\item[i)] The closed-loop system \eqref{eq:Closed-loop2} is ISS with nonlinear stability margin
		\begin{equation}\label{eq:nonlinearsm}
			\rho(\norm{\bar{x}}):=g_m\norm{\bar{x}},
			\end{equation}
			with $g_m\in\rea_+$ being the \textit{gain margin} defined as
			\begin{equation}\label{eq:gm}
				g_m:=\dfrac{\mu\beta_{\max}^2\theta}{\maxlamb{K_p}}
			\end{equation}
			for $0<\theta<1$.
			\item[ii)] The trajectories of the closed-loop system \eqref{eq:Closed-loop2} converge exponentially to 
			$$\Omega_e:=\left\{\bar{x}\in\rea^{3n}|\norm{\bar{x}}=\dfrac{1}{g_m}\norm{d_u}\right\}$$ at the rate of convergence upper bounded by
			\begin{equation}\label{eq:rate2}
				\dfrac{\mu\beta_{\max}(1-\theta)}{1+\epsilon\beta_{\max}\maxlamb{{M}_d}};
			\end{equation}
			and the maximum overshoot of the output of the system is given by \eqref{eq:ov}.
	\end{itemize}
\end{Theorem}

\begin{proof}
	To prove i), consider the bounds \eqref{eq:cond1_exp}; it follows that, via some computations, the derivative of the Lyapunov candidate corresponds to	
	\begin{equation*}
	\dot{S}=-\nabla^\top \bar{H} \Upsilon \nabla \bar{H}+ d_u^\top{K}_p\bar{q}
	\end{equation*}
	with $\Upsilon(\bar{q},\bar{p})$ as defined as in \eqref{eq:Upsilon}.
	
	Denote with $\mu\in\rea_+$ the minimum eigenvalue of \eqref{eq:Upsilon}, then we get that
	\begin{equation*}
		\dot{{S}}\leq-\mu\norm{\nabla\hat{H}}^2+\maxlamb{K_p}\norm{d_u}\norm{\bar{x}}.
	\end{equation*}
	Consider $0<\theta<1$ and the bounds \eqref{eq:boundsHamiltonian}, then we can rewrite the previous equation as
	\begin{equation*}
		\begin{array}{rcl}
			\dot{S}&\leq&-\mu\beta_{\max}^2\norm{\bar{x}}^2(1-\theta)+\maxlamb{K_p}\norm{d_u}\norm{\bar{x}}\\
			&&-\mu\beta_{\max}^2\theta\norm{\bar{x}}^2.
		\end{array}
	\end{equation*}
	Hence,
	\begin{equation}\label{eq:SdotISS}
		\dot{S}\leq -\mu\beta_{\max}^2\norm{\bar{x}}^2(1-\theta),~\forall~ \norm{\bar{x}}\in\Omega,
	\end{equation}
	with $\Omega:=\left\{\bar{x}\in\rea^{3n}|\norm{d_u}\leq \rho(\norm{\bar{x}})\right\}$, and $\rho(\norm{\bar{x}})$ is defined as in \eqref{eq:nonlinearsm}.
	
	Thus, from \eqref{eq:cond1_exp} and \eqref{eq:SdotISS},  the closed-loop system \eqref{eq:Closed-loop2} is ISS with nonlinear stability margin $\rho(\norm{\bar{x}})$ (see \cite{sontag1995characterizations} and Theorem 4.19 from \cite{khalil2002nonlinear}).
	
	To prove ii), from \eqref{eq:cond1_exp} and \eqref{eq:SdotISS}, we get that
	\begin{equation*}
		\dot{S}\leq- \dfrac{2\mu\beta_{\max}(1-\theta)}{1+\epsilon\beta_{\max}\maxlamb{{M}_d}}S;
	\end{equation*}
	then, by applying the comparison lemma~--~similar as in the proof of Theorem \ref{Prop1}~--~we get that the upper bound of the rate of convergence is given by \eqref{eq:rate2}, and the maximum overshoot of the output of the system is given by \eqref{eq:ov}.
\end{proof}
	~\\
	\begin{remark}
		Similar to Section \ref{sec: PBIC_approach}, we can exploit \eqref{eq:nonlinearsm}, \eqref{eq:gm}, \eqref{eq:rate2}, and \eqref{eq:ov} as tuning guidelines, since these expressions highlight the effects of the control parameters on the nonlinear gain margin, the upper bound of the rate of convergence, and the maximum permissible overshoot of the closed-loop system subjected to unmatched disturbances. Additionally, note that the gain margin \eqref{eq:gm} is the maximum growth of the norm of the disturbance with respect to the norm of the trajectories in which the closed-loop system remains ISS; that is, the closed-loop is more robust against unmatched disturbances for a higher gain margin.
		\hfill$\square$
	\end{remark}

\section{Experiments}\label{sec:exp}
This section illustrates the effectiveness of our control law \eqref{eq:Control_Law} to reject disturbances while prescribing the desired performance in terms of the upper bound of the rate of convergence and the maximum overshoot of the system. To this end, we employ the Philips Experimental Robotic Arm (PERA)~--~as depicted in Fig.\ref{pera}~--~which is a seven-DoF experimental robotic arm designed by Philips Applied Technologies \cite{rijs2010philips} to
mimic the motion of a human arm. To ease the presentation of our results, we reduce the model to three-DoF, namely, shoulder roll ($q_1$), elbow pitch ($q_2$), and elbow roll ($q_3$). The model of the PERA is given by \eqref{eq:sigma_PH_mechanical} with $n=3$,  $D(q,p)=0_{3\times 3}$, $G=I_3$, and we omit the  expressions for $M(q)$ and $U(q)$ due to space constraints. We refer the reader to \cite{ariasenergy} for further details; additionally, a MATLAB\textsuperscript{\textregistered} script to generate  $M(q)$ and $U(q)$ can be found in \cite{bol2}. 
\begin{figure}[t]
	\centering
	\includegraphics[width=0.6\columnwidth]{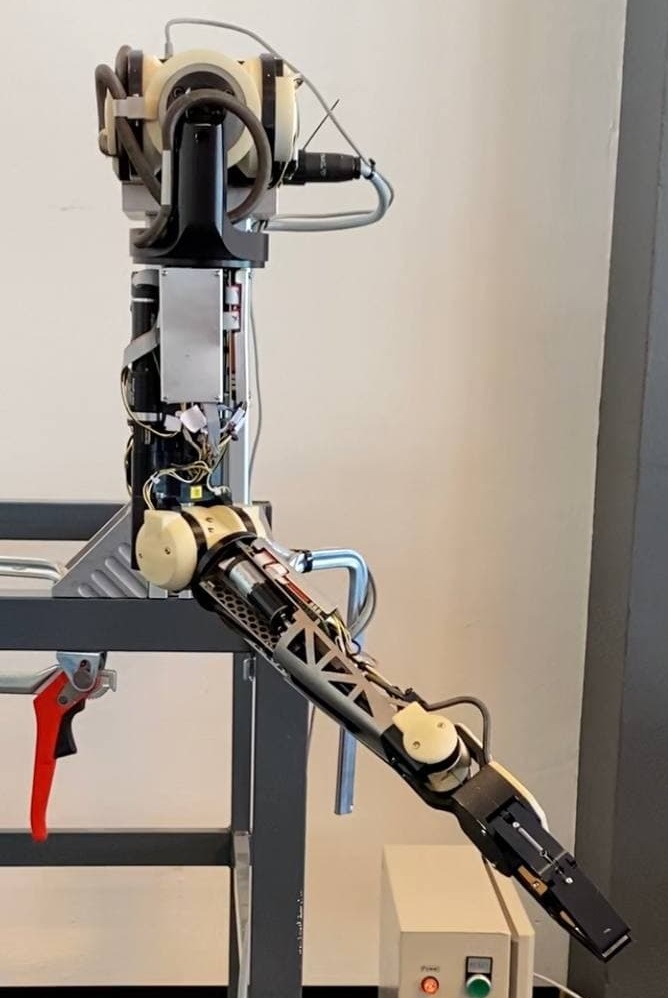}
	\caption{The PERA setup}\label{pera}
\end{figure}

To compare against our PBIC approach \eqref{eq:Control_Law} with a standard approach, we employ a general PBC methodology without an integrator, i.e.,\footnote{This control structure is generally known as energy shaping and damping injection controller. Such methodology can be found in  \cite{borja2020new},\cite{gomez2004physical},\cite{ortega2002stabilization}.} 
\begin{equation}\label{eq:controlcomp}
	u=-K_{\tt{es}}( q-q_\star)-K_{\tt{di}} \dot{q}.
\end{equation}
Table \ref{gains} shows the gain selection for \eqref{eq:controlcomp}~--~depicted as {``Case 1"}~--~and the gains selection for our controller \eqref{eq:Control_Law}~--~depicted as ``Case 2" and ``Case 3". The gain selection for these last two cases verify \eqref{eq:cond1}.
\begin{table}[t]
	\centering
	\caption{Controllers gains}\label{gains}
	\begin{tabular}{cccc}
		\hline
		& Case 1           & Case 2              & Case 3                 \\ \hline
		$K_{es}$ & diag\{75,50,50\} & -                   & -                      \\
		$K_{di}$ & diag\{7,5,5\}    & -                   & -                      \\
		$K_p$  & -                & \multicolumn{2}{c}{diag\{10,7.5,7.5\}}       \\
		$K_i$  & -                & \multicolumn{2}{c}{diag\{15,10,10\}}         \\
		$K_d$  & -                & \multicolumn{2}{c}{diag\{7,5,5\}}            \\
		$M_d$  & -                & diag\{0.2,0.2,0.2\} & diag\{0.06,0.06,0.06\} \\ \hline
	\end{tabular}
\end{table}
\begin{figure}[t]
	\centering
	\includegraphics[width=\columnwidth]{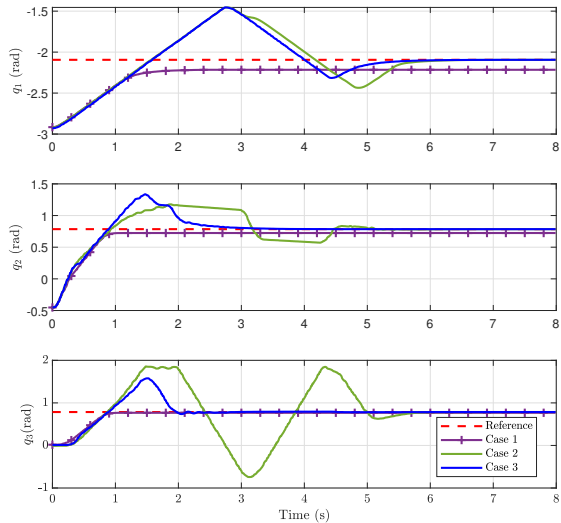}
	\caption{Trajectories of the angular positions}\label{exp}
\end{figure}

The results for each case in Table \ref{gains} are depicted in Fig.\ref{exp} and a video with the experimental results can be found in: \url{https://youtu.be/eussb6kHhQc}.

Note that there is a steady-state error in the joint positions for ``Case 1," which trajectories are obtained with controller \eqref{eq:controlcomp}. This error may be due to unmodeled physical phenomena (e.g., dry friction or asymmetry of the motors) or noise in the measurement sensors. Then, the steady-state error is removed by the implementation of our controller \eqref{eq:Control_Law} for ``Case 2" (or in``Case 3").

We apply \eqref{eq:rate2} to improve the upper bound of the rate of convergence, which is directly proportional to $\beta_{\max}$.  Recall that $${\beta_{\max}:=\max\{\maxlamb{\inv{M_d}},\maxlamb{K_p},\maxlamb{\inv{K_i}}\}}.$$ 
It follows that $\beta_{\max}=10$ for ``Case 2'', and by increasing $\maxlamb{\inv{M_d}}$  in ``Case 3'', we get that ${\beta_{\max}=16.67}$. Thus, as expected, the rate of convergence increases in ``Case 3" with respect to ``Case 2".  However, from tuning rule \eqref{eq:ov}, note that the increment of the rate of convergence comes at the expense of a higher overshoot in the trajectories, as verified particularly for joints $q_1$ and $q_2$.

\section{Conclusions and Future work}\label{sec:concl}
We have shown a PBIC approach that exponentially stabilizes a fully actuated mechanical system in the desired equilibrium even at the presence of disturbances. In particular, the latter is achieved partly by exploiting the gyroscopic forces contained in the mechanical system, which in literature usually are canceled via feedback linearization. Moreover, we have also provided tuning rules to prescribe a performance in terms of the rate of convergence, the maximum overshoot of the output, and the nonlinear stability margin. Additionally, we have verified the effectiveness of the approach and its tuning rules in a three-DoF robotic arm. 

Regarding future work, we aim to develop further tuning rules for our PBIC approach to prescribe the behavior in the vicinity of the closed-loop system in terms of oscillations, rise time, or settling time. Moreover, we aim to adapt our methodology into an underactuated mechanical setting.

\bibliographystyle{ieeetr}
\bibliography{ref.bib} 
\end{document}